%% file: draft.tex
\documentclass[12pt]{amsart}
\usepackage[varg]{txfonts}
\usepackage{amsmath}
\usepackage[cm]{fullpage}
\usepackage{hyperref}
\usepackage{cite}
\newtheorem{theorem}{Theorem}

\newtheorem{claim}[theorem]{Claim}
\newtheorem{fact}[theorem]{Claim}

\usepackage{tikz}
\usetikzlibrary{calc}

\title{A recognition algorithm for simple-triangle graphs}
\author{Asahi Takaoka}
\address{
  Department of Information Systems Creation, 
  Kanagawa University, 
  Rokkakubashi 3-27-1 Kanagawa-ku, 
  Kanagawa, 221--8686, Japan 
}
\email{takaoka@kanagawa-u.ac.jp}
\date{\today}



\begin{document}

\input{main}


\input{draft.bbl}
\end{document}

%% file: main.tex
\begin{abstract}
A simple-triangle graph is the intersection graph of 
triangles that are defined by a point on a horizontal line and 
an interval on another horizontal line. 
The time complexity of the recognition problem for simple-triangle graphs 
was a longstanding open problem, which was recently settled. 
This paper provides a new recognition algorithm for simple-triangle graphs 
to improve the time bound from $O(n^2 \overline{m})$ to $O(nm)$, 
where $n$, $m$, and $\overline{m}$ are 
the number of vertices, edges, and non-edges of the graph, respectively. 
The algorithm uses the vertex ordering characterization 
that a graph is a simple-triangle graph if and only if 
there is a linear ordering of the vertices 
containing both an alternating orientation of the graph and 
a transitive orientation of the complement of the graph. 
We also show, as a byproduct, that 
an alternating orientation can be obtained in $O(nm)$ time for cocomparability graphs, 
and it is NP-complete to decide whether a graph has 
an orientation that is alternating and acyclic. 
\end{abstract}
\maketitle

\section{Introduction}
A graph is an \emph{intersection graph} 
if there is a set of objects such that 
each vertex corresponds to an object 
and two vertices are adjacent if and only if 
the corresponding objects have a nonempty intersection. 
Such a set of objects is a \emph{representation} of the graph. 
See~\cite{BLS99,Golumbic04,MM99,Spinrad03} for survey. 
Let $L_1$ and $L_2$ be two horizontal lines in the plane with $L_1$ above $L_2$. 
\emph{Trapezoid graphs} are the intersection graphs of
trapezoids that are defined by an interval on $L_1$ and an interval on $L_2$. 
Trapezoid graphs have been introduced in~\cite{CK87-CN,DGP88-DAM} 
as a generalization of both interval graphs and permutation graphs. 
Many recognition algorithms and structural characterizations are provided 
for trapezoid graphs~\cite{CC96-DAM,Cogis82b-DM,GT04,FHM94-SIAMDM,Langley95-DAM,MS94-JAL,MC11-DAM}. 

We obtain some interesting subclasses of trapezoid graphs 
by restricting the trapezoids in the representation. 
A trapezoid graph is a \emph{simple-triangle graph} 
if every trapezoid in the representation is a triangle 
with the apex on $L_1$ and the base on $L_2$. 
Similarly, a trapezoid graph is a \emph{triangle graph} 
if every trapezoid in the representation is a triangle, 
but there is no restriction on which line contains the apex and the base. 
Simple-triangle graphs and triangle graphs have been introduced in~\cite{CK87-CN} 
and studied under the name of \emph{PI graphs} and \emph{PI* graphs}, 
respectively~\cite{BLS99,COS08-ENDM,CK87-CN,Spinrad03}, 
where \emph{PI} stands for \emph{Point-Interval}. 
Moreover, a trapezoid graph is a \emph{parallelogram graph} 
if every trapezoid is a parallelogram, and 
parallelogram graphs coincide with \emph{bounded tolerance graphs}~\cite{BFIL95-DAM,GMT84-DAM}. 
These three graph classes are proper subclasses of trapezoid graphs, 
and they contain both interval graphs and permutation graphs 
as proper subclasses. 

The time complexity of the recognition problem of these three graph classes 
was a longstanding open problem~\cite{BLS99,Spinrad03}, 
which was recently settled. 
While the problem is NP-complete for triangle graphs~\cite{Mertzios12-TCS} 
and parallelogram graphs~\cite{MSZ11-SIAMCOMP}, 
an $O(n^2 \overline{m})$-time recognition algorithm 
has been given for simple-triangle graphs~\cite{Mertzios15-SIAMDM}, 
where $n$ and $\overline{m}$ are the number of vertices and non-edges 
of the graph, respectively. 
This algorithm reduces the recognition problem of simple-triangle graph to 
the \emph{linear-interval cover problem}, a problem of covering an associated 
bipartite graph by two chain graphs satisfying additional conditions. 
In~\cite{Takaoka17-LNCS}, we showed an alternative algorithm 
for the linear-interval cover problem, 
but 
this does not improve the running time. 
Meanwhile, we showed in~\cite{Takaoka18-DM} 
a vertex ordering characterization that 
a graph is a simple-triangle graph if and only if 
there is a linear ordering of the vertices 
containing both an alternating orientation of the graph and 
a transitive orientation of the complement of the graph. 
Using this vertex ordering characterization, we show in this paper 
an $O(nm)$-time recognition algorithm for simple-triangle graphs, 
where $m$ is the number of edges of the graph. 

Our algorithm is shown in the next section, and 
correctness of the algorithm is proved in Section~\ref{section:correctness}. 
We finally discuss our results and 
further research in Section~\ref{section:conclusion}.

\section{The recognition algorithm}\label{section:algorithm}
\subsection{Preliminaries}

\paragraph{Notation}
In this paper, we will deal only with finite graphs 
having no loops and multiple edges. 
Unless stated otherwise, graphs are assumed to be undirected, 
but we also deal with graphs having directed edges. 
We write $uv$ for the \emph{undirected edge} 
joining a vertex $u$ and a vertex $v$, and 
we write $(u, v)$ for the \emph{directed edge} 
from $u$ to $v$. 
For a graph $G = (V, E)$, we sometimes write $V(G)$ for 
the vertex set $V$ of $G$ and write $E(G)$ for the edge set $E$ of $G$. 

Let $G = (V, E)$ be an undirected graph. 
The \emph{complement} of $G$ 
is the graph $\overline{G} = (V, \overline{E})$ 
such that $uv \in \overline{E}$ if and only if $uv \notin E$ 
for any two vertices $u, v \in V$. 
A sequence of distinct vertices $(v_0, v_1, \dots, v_k)$ is 
a \emph{path from $v_0$ to $v_k$} in $G$ 
if $v_0v_1, v_1v_2, \dots, v_{k-1}v_k \in E$. 
These edges are \emph{the edges on the path}. 
The \emph{length} of the path is the number $k$ of the edges 
on the path. 
A path $(v_0, v_1, \dots, v_k)$ is a \emph{cycle} in $G$ 
if in addition $v_kv_0 \in E$. 
The edges $v_0v_1, v_1v_2, \dots, v_{k-1}v_k$ and $v_kv_0$ 
are \emph{the edges on the cycle}. 
The \emph{length} of the cycle is the number $k+1$ of the edges 
on the cycle. 
A \emph{chord} of a path (cycle) is an edge joining two vertices 
that are not consecutive on the path (cycle). 
A path (cycle) is \emph{chordless} if it contains no chords. 

Let $G = (V, E)$ be an undirected graph. 
An \emph{orientation} of $G$ is a graph obtained from $G$ 
by orienting each edge in $E$, that is, 
replacing each edge $uv \in E$ with either $(u, v)$ or $(v, u)$. 
An \emph{oriented graph} is an orientation of some graph. 
A \emph{partial orientation} of $G$ is a graph obtained from $G$ 
by orienting each edge in a subset of $E$. 
A \emph{partially oriented graph} is 
a partial orientation of some graph. 
Notice that a (partially) oriented graph contains 
no pair of edges $(u, v)$ and $(v, u)$ for some vertices $u, v$. 
We will denote a (partial) orientation of a graph 
only by its edge set when the vertex set is clear from the context. 

Let $H = (V, F)$ be a (partially) oriented graph. 
A sequence of distinct vertices $(v_0, v_1, \dots, v_k)$ 
is a \emph{directed path from $v_0$ to $v_k$} in $H$ 
if $(v_0, v_1), (v_1, v_2), \dots, (v_{k-1}, v_k) \in F$. 
A directed path $(v_0, v_1, \dots, v_k)$ is 
a \emph{directed cycle} in $H$ if in addition $(v_k, v_0) \in F$. 
The edges on the path (cycle) and 
the length of the path (cycle) are 
defined analogously to the undirected case. 
A (partial) orientation of a graph is called \emph{acyclic} 
if it contains no directed cycles. 
A \emph{linear extension} (or topological sort) of 
an acyclic (partial) orientation $F$ is 
a linear ordering of the vertices such that 
$(u, v) \in F$ if and only if $u$ is a predecessor of $v$ in the ordering. 
Let $F_s \subseteq F$ be a set of directed edges in $H$. 
The \emph{reversal} $F_s^{-1}$ of $F_s$ is the set of directed edges 
obtained from $F_s$ by reversing all the edges in $F_s$, that is, 
$F_s^{-1} = \{ (u, v)\colon\ (v, u) \in F_s \}$.

\paragraph{Comparability graphs}
An orientation $F$ of a graph $G$ is a \emph{transitive orientation} 
if $(u, v) \in F$ and $(v, w) \in F$ then $(u, w) \in F$. 
A graph is a \emph{comparability graph} 
if it has a transitive orientation. 
The complement of a comparability graph is a \emph{cocomparability graph}. 
The class of cocomparability graphs contains the class of trapezoid graphs 
as a proper subclass~\cite{DGP88-DAM,CK87-CN}, and hence 
the complement of any simple-triangle graph has a transitive orientation. 
Note that every cocomparability graph contains no chordless cycle of length 
greater than or equal to 5 (see~\cite{BLS99,Gallai67} for example), and thus 
every chordless cycle of a cocomparability graph has length at most 4. 

An orientation $F$ of a graph $G$ is \emph{quasi-transitive} 
if $(u, v) \in F$ and $(v, w) \in F$ then $uw \in E(G)$ and therefore either $(u, w) \in F$ or $(w, u) \in F$. 
In other words, an orientation $F$ of $G$ is quasi-transitive 
if for any path of three vertices $(u, v, w)$ in $G$, 
either $(u, v), (w, v) \in F$ or $(v, u), (v, w) \in F$. 
We can see that an orientation is transitive if and only if 
it is quasi-transitive and acyclic. 
Trivially, a graph having a transitive orientation 
also has a quasi-transitive orientation. 
The converse is also known to be true~\cite{Ghouila-Houri62,HH95-JGT}, 
that is, if a graph $G$ has a quasi-transitive orientation $F$ that is not acyclic, 
then $G$ has another quasi-transitive orientation $F'$ that is also acyclic. 

There is a well-known algorithm for 
recognizing comparability graphs and 
producing a transitive orientation of a graph, 
which takes $O(nm)$ time 
(see~\cite{Golumbic04} for example). 
A linear-time algorithm is also known for producing 
a transitive orientation of a comparability graph~\cite{MS99-DM}. 
To be precise, the algorithm produces a linear extension 
of the transitive orientation. 
However, even if the given graph is not a comparability graph, 
the linear-time algorithm may produce an orientation that is not transitive. 
Hence, to recognize comparability graphs, 
we must verify transitivity of the orientation. 
The best known method for verifying transitivity uses 
Boolean matrix multiplication to test whether $A^2 = A$, 
where $A$ is the adjacency matrix of the orientation $F$ 
(see~\cite{Spinrad03} for example). 
An alternative method for verifying transitivity, 
which takes $O(m^{3/2})$ time, is discussed in~\cite[Section~11.1.4]{Spinrad03}. 

Let $\overline{G}$ be the complement of a graph $G$. 
The linear-time algorithm of~\cite{MS99-DM} also can produce a linear extension 
of an orientation $\overline{F}$ of $\overline{G}$ 
such that $\overline{F}$ is transitive if and only if 
$G$ is a cocomparability graph. 
We show a method for verifying transitivity of $\overline{F}$ in $O(nm)$ time. 
Recall that an orientation is transitive if and only if 
it is quasi-transitive and acyclic. 
Obviously, $\overline{F}$ is acyclic. 
We can test whether $\overline{F}$ is quasi-transitive by checking 
for any path of three vertices $(u, v, w)$ in $\overline{G}$, 
either $(u, v), (w, v) \in \overline{F}$ or $(v, u), (v, w) \in \overline{F}$. 
The number of paths in $G$ of three vertices is at most $nm$, 
and they can be found in $O(nm)$ time. 
Thus we have the following. 
\begin{theorem}\label{theorem:cocomparability}
Cocomparability graphs can be recognized in $O(nm)$ time. 
\end{theorem}

\paragraph{Alternately orientable graphs}
An orientation of a graph is \emph{alternating} 
if it is (quasi-)transitive on every chordless cycle of length 
greater than or equal to 4, that is, the directions of the edges alternate. 
A graph is an \emph{alternately orientable graph}~\cite{Hoang87-JCTSB} 
if it has an alternating orientation. 
It is clear from the definition that 
alternately orientable graphs generalize comparability graphs. 
A polynomial-time recognition algorithm is known 
for alternately orientable graphs~\cite{Hoang87-JCTSB}. 

We say a graph is an \emph{alternately orientable cocomparability graph} 
if it is a cocomparability graph that has an alternating orientation. 
Obviously, alternately orientable cocomparability graphs 
can be recognized in polynomial time. 
It is also known that 
every alternately orientable cocomparability graph is a trapezoid graph 
but the converse is not true~\cite{Felsner98-JGT}. 

We say an orientation is an \emph{acyclic alternating orientation} 
if it is alternating and acyclic, and 
we say a graph is \emph{acyclic alternately orientable} 
if it has an acyclic alternating orientation. 
We can see that an alternating orientation is acyclic 
if it contains no directed cycles of length 3. 
Recall that a graph has a quasi-transitive orientation 
if and only if 
it has an orientation that is quasi-transitive and acyclic. 
It was conjectured that 
a similar statement might hold for alternating orientation, 
that is, a graph is alternately orientable if and only if 
it has an acyclic alternating orientation~\cite{Hoang87-JCTSB}. 
Later, however, a counterexample was provided~\cite{Hougardy93-DM}. 
Thus, the class of acyclic alternately orientable graphs 
is a proper subclass of alternately orientable graphs. 
Moreover, the recognition problem is NP-complete 
for acyclic alternately orientable graphs. 
\begin{theorem}\label{theorem:acyclic-alternately-orientable}
The recognition problem is NP-complete 
for acyclic alternately orientable graphs. 
\end{theorem}
\begin{proof}
We can verify in polynomial time whether 
an orientation is alternating~\cite{Hoang87-JCTSB}. 
Since testing acyclicity of an orientation takes linear time, 
the recognition problem is in NP. 
We show a polynomial-time reduction 
from the non-betweenness problem, 
which is known to be NP-complete~\cite{GM06-LNCS}. 
Given a finite set $A$ and a collection $C$ of 
ordered triples $(a_i, a_j, a_k)$ of distinct elements of $A$, 
the non-betweenness problem is to decide 
whether there is a bijection $f: A \to \{1, 2, \ldots, |A|\}$ 
such that for each $(a_i, a_j, a_k) \in C$, 
either $f(a_i), f(a_k) < f(a_j)$ or $f(a_j) < f(a_i), f(a_k)$. 
Let $G$ be a graph such that 
\begin{align*}
V(G) &= \{ v_i \colon\ a_i \in A \} \cup \{ u_i, w_i \colon\ t_i \in C \} \mbox{ and } \\
E(G) &= \{ v_iv_j \colon\ 1 \leq i < j \leq |A| \} \cup \{ u_hw_h, u_hv_j, w_hv_i, w_hv_k \colon\ t_h = (a_i, a_j, a_k) \in C \}.
\end{align*}
Clearly, we can construct the graph $G$ in time polynomial in $n$ and $m$. 
We can also see from the construction that 
the set of chordless cycles 
$\{(u_h, w_h, v_i, v_j), (u_h, w_h, v_k, v_j) \colon\ t_h = (a_i, a_j, a_k) \in C \}$ 
contains all the chordless cycles of length grater than or equal to 4. Thus 
an orientation $F$ of $G$ is alternating if and only if 
for any three vertices $v_i, v_j, v_k$ with 
$t_h = (a_i, a_j, a_k) \in C$, either 
$(v_i, v_j), (v_k, v_j) \in F$ or $(v_j, v_i), (v_j, v_k) \in F$. 
Therefore, $G$ has an acyclic alternating orientation if and only if there is 
a bijection $f: A \to \{1, 2, \ldots, |A|\}$ 
such that for each $(a_i, a_j, a_k) \in C$, 
either $f(a_i), f(a_k) < f(a_j)$ or $f(a_j) < f(a_i), f(a_k)$. 
Thus we have the theorem
\end{proof}

\paragraph{The vertex ordering characterization of simple-triangle graphs}
Now, we show the vertex ordering characterization 
of simple-triangle graphs. 
\begin{theorem}[\cite{Takaoka18-DM}]
\label{theorem:VOC}
A graph $G$ is a simple-triangle graph if and only if 
there is an (acyclic) alternating orientation $F$ of $G$ and 
a transitive orientation $\overline{F}$ of 
the complement $\overline{G}$ of $G$ such that 
$F \cup \overline{F}$ is acyclic. 
Moreover, if $G$ is a simple-triangle graph, 
then for any transitive orientation $\overline{F}$ of $\overline{G}$, 
there is an (acyclic) alternating orientation $F$ of $G$ 
such that $F \cup \overline{F}$ is acyclic. 
An acyclic orientation of a complete graph is equivalent 
to the linear ordering of the vertices of the graph. 
The orientation $F \cup \overline{F}$ is 
called an \emph{apex ordering} of a simple-triangle graph 
since it coincides with the ordering of 
the apices of the triangles in the representation. 
\end{theorem}
Theorem~\ref{theorem:VOC} implies that a simple-triangle graph 
is an (acyclic) alternately orientable cocomparability graph, 
but the converse is not known to be true~\cite{Takaoka18-DM}. 
We also note that Theorem~\ref{theorem:VOC} is similar to 
a well-known fact that a graph is a permutation graph if and only if 
it is both a comparability graph and a cocomparability graph~\cite{PLE71-CJM}.

\subsection{The recognition algorithm}
Let $F$ be an alternating orientation of a graph $G$, and 
let $\overline{F}$ be a transitive orientation of the complement $\overline{G}$ of $G$. 
Suppose that $F \cup \overline{F}$ is not acyclic. 
It is well known that an orientation of a complete graph is not acyclic 
if and only if it contains a directed cycle of length 3. 
Each directed cycle of length 3 in $F \cup \overline{F}$ consists of 
either three edges in $F$ or two edges in $F$ with one edge in $\overline{F}$ 
since $\overline{F}$ is a transitive orientation. 
We will refer to a directed cycle $(a, b, c)$ of length 3 in $F \cup \overline{F}$ 
as a \emph{$\Delta$-obstruction} 
if $(a, b), (b, c) \in F$ and $(c, a) \in \overline{F}$. 

It is clear from Theorem~\ref{theorem:VOC} that 
if a graph $G$ is a simple-triangle graph, 
then there is a pair of 
an acyclic alternating orientation $F$ of $G$ 
and a transitive orientation $\overline{F}$ of $\overline{G}$ 
such that $F \cup \overline{F}$ contains no $\Delta$-obstructions. 
Conversely, a graph $G$ is a simple-triangle graph 
if for some transitive orientation $\overline{F}$ of $\overline{G}$, 
there is an acyclic alternating orientation $F$ of $G$ 
such that $F \cup \overline{F}$ contains no $\Delta$-obstructions. 
We can also have the following. 
\begin{theorem}\label{theorem:correctness}
Let $\overline{F}$ be a transitive orientation of $\overline{G}$. 
If $G$ has a (not necessarily acyclic) alternating orientation $F$ 
such that $F \cup \overline{F}$ contains no $\Delta$-obstructions, 
then $G$ also has an acyclic alternating orientation $F'$ of $G$ 
such that $F' \cup \overline{F}$ contains no $\Delta$-obstructions, 
that is, $G$ is a simple-triangle graph. 
\end{theorem}
Our recognition algorithm is due to this structural characterization, 
and we will prove the theorem when we show the correctness of the algorithm. 

Figure~\ref{algorithm} shows our algorithm to recognize simple-triangle graphs. 
The algorithm finds an apex ordering of the given graph 
if it is a simple-triangle graph 
or report that the graph is not a simple-triangle graph. 
\begin{figure*}[ht]
\centering
\fbox{
\begin{tabular}{lp{38em}}
  \textbf{Input:}&
  A graph $G$. 
  \\
  \textbf{Output:}&
  An apex ordering of $G$ if $G$ is a simple-triangle graph or 
  \\ &
  report that $G$ is not a simple-triangle graph. 
  \\
  \textbf{Step~1:}&
  Compute a transitive orientation $\overline{F}$ of the complement $\overline{G}$ of $G$. 
  \\ &
  If $G$ has no transitive orientations, 
  then report that $G$ is not a simple-triangle graph. 
  \\
  \textbf{Step~2:}&
  Compute a partial orientation $F$ of $G$ that satisfies the following three conditions: 
  \begin{enumerate}
  \item $F \cup \overline{F}$ contains no $\Delta$-obstructions, 
  \item The edges on every chordless cycle of length 4 are oriented so that 
  the directions of the edges alternate 
  (Notice that the length of every chordless cycle in $G$ is at most 4 
  since $G$ is a cocomparability graph), and 
  \item Every edge remains undirected if it is not on a chordless cycle of length 4. 
  \end{enumerate}
  If $G$ has no such orientations, 
  then report that $G$ is not a simple-triangle graph. 
  \\
  \textbf{Step~3:}&
  Choose a vertex $v$ of $G$. 
  Let $F_v$ be the set of all the  edges $(w, u) \in F$ 
  such that $(u, v, w)$ form a directed cycle in $F$. 
  Reverse the orientation of all the edges in $F_v$ 
  to compute another partial orientation $F'$ of $G$, that is, 
  $F' = (F - F_v) \cup F_v^{-1}$. 
  \\ &
  Repeat this procedure for all the vertices of $G$. 
  We denote the resultant orientation by $F''$. 
  \\ 
  \textbf{Step~4:}&
  Output a linear extension of $F'' \cup \overline{F}$. 
\end{tabular}
}
  \caption{The recognition algorithm for simple-triangle graphs}
  \label{algorithm}
\end{figure*}

We remark that in \textbf{Step~2} of the algorithm, 
we keep an edge undirected if it is not on a chordless cycle of length 4, 
since it is needed to do \textbf{Step~3} and \textbf{Step~4} correctly. 
If we do not need to find an apex ordering of the given graph, 
in \textbf{Step~2} 
we only have to compute an alternating orientation $F$ of $G$ such that 
$F \cup \overline{F}$ contains no $\Delta$-obstructions. 

We will prove the correctness of the algorithm in the next section. 
Notice that in order to prove the correctness, 
it suffices to show that 
when \textbf{Step~3} is finished, $F'' \cup \overline{F}$ is acyclic.

\paragraph{Details of the algorithm}
In the rest of this section, 
we show that the algorithm runs in $O(nm)$ time. 
We assume without loss of generality that 
the given graph is connected and $n \leq m$, 
since otherwise we apply the algorithm to each connected component. 

In \textbf{Step~1}, 
we use the linear-time algorithm of~\cite{MS99-DM}, which produces 
a linear extension of a transitive orientation $\overline{F}$ of 
the complement $\overline{G}$ 
if the given graph $G$ is a cocomparability graph. 
As shown in Theorem~\ref{theorem:cocomparability}, 
we can verify transitivity of $\overline{F}$ in $O(nm)$ time, 
and hence \textbf{Step~1} can be performed in the same time bound. 

In \textbf{Step~2}, we construct the auxiliary graph $G^+$ 
of $G$ that is bipartite if and only if $G$ is alternately orientable. 
Then, we make a 2CNF formula $\phi$, 
an instance of the 2-satisfiability problem, from $G^+$. 
The partial orientation $F$ of $G$ can be obtained from 
a truth assignment $\tau$ of $\phi$. 
When $G^+$ is not bipartite or $\phi$ cannot be satisfied, 
the algorithm reports that $G$ is not a simple-triangle graph. 

The vertices of the auxiliary graph $G^+$ are 
all the ordered pairs $(u, v)$ with $uv \in E(G)$. 
Each vertex $(u, v)$ of $G^+$ is adjacent to the vertex $(v, u)$. 
Each vertex $(u, v)$ of $G^+$ is also adjacent to 
every vertex $(v, w)$ such that in $G$, 
the vertices $(u, v, w)$ form a path of three vertices on a chordless cycle of length 4. 
We can see that $G$ is alternately orientable if and only if 
$G^+$ is bipartite, and an alternating orientation of $G$ 
can be obtained from one color class of $G^+$. 
The number of vertices of $G^+$ is $2m$, 
and the number of edges of $G^+$ is at most $m+2nm$ 
since the number of paths of three vertices in $G$ is at most $nm$. 
Hence we can test the bipartiteness of $G^+$ in $O(nm)$ time. 

Now, we show that the auxiliary graph $G^+$ can be constructed from $G$ in $O(nm)$ time. 
The \emph{neighborhood} of a vertex $v$ of $G$ is the set $N(v) = \{u \in V(G) \colon\ uv \in E(G)\}$. 
We define that 
the \emph{upper set} of $v$ is the set $U(v) = \{u \in V(G) \colon\ (v, u) \in \overline{F}\}$ and 
the \emph{lower set} of $v$ is the set $L(v) = \{u \in V(G) \colon\ (u, v) \in \overline{F}\}$. 
Let $uv$ be an edge of $G$. 
If $N(u) \cap U(v) \neq \emptyset$ and $N(v) \cap U(u) \neq \emptyset$, then 
for any two vertices $w \in N(u) \cap U(v)$ and $z \in N(v) \cap U(u)$, we have $wz \in E(G)$, 
that is, $(u, v, z, w)$ is a chordless cycle of length 4; 
otherwise either $(w, z) \in \overline{F}$ or $(z, w) \in \overline{F}$, but 
if $(w, z) \in \overline{F}$ then 
$(v, w) \in \overline{F}$ implies $(v, z) \in \overline{F}$, contradicting $z \in N(v)$, and 
if $(z, w) \in \overline{F}$ then 
$(u, z) \in \overline{F}$ implies $(u, w) \in \overline{F}$, contradicting $w \in N(u)$. 
Similarly, 
for any two vertices $w \in N(u) \cap L(v)$ and $z \in N(v) \cap L(u)$, 
we have that $(u, v, z, w)$ is a chordless cycle of length 4. 
Conversely, each chordless cycle of length 4 has two edges 
$uv$ and $zw$ such that 
$N(u) \cap U(v) \neq \emptyset$, $N(v) \cap U(u) \neq \emptyset$, 
$N(z) \cap L(w) \neq \emptyset$, and $N(w) \cap L(z) \neq \emptyset$. 
\par
Therefore, we can construct the auxiliary graph $G^+$ by the following method. 
(1) For each edge $uv \in E(G)$, test whether 
$N(u) \cap U(v) \neq \emptyset$ and $N(v) \cap U(u) \neq \emptyset$ 
[resp. $N(u) \cap L(v) \neq \emptyset$ and $N(v) \cap L(u) \neq \emptyset$]. 
(2) If it is, then for any two vertices 
$w \in N(u) \cap U(v)$ and $z \in N(v) \cap U(u)$ 
[resp. $w \in N(u) \cap L(v)$ and $z \in N(v) \cap L(u)$], 
add to $G^+$ 
the edge joining $(u, v)$ and $(v, z)$, 
the edge joining $(z, v)$ and $(v, u)$, 
the edge joining $(v, u)$ and $(u, w)$, and 
the edge joining $(w, u)$ and $(u, v)$. 
(3) Finally, for each edge $uv \in E(G)$, 
add to $G^+$ the edge joining $(u, v)$ and $(v, u)$. 
The first procedure takes $O(n)$ time for each edge, and 
the second procedure takes $O(nm)$ time in total 
since the number of paths of three vertices in $G$ is at most $nm$. 
The third procedure takes $O(m)$ time in total, and hence 
the auxiliary graph $G^+$ can be constructed in $O(nm)$ time. 

We next construct the 2CNF formula $\phi$ from $G^+$. 
Recall that each vertex $(u, v)$ of $G^+$ is adjacent to $(v, u)$. 
A connected component of $G^+$ consists of only two vertices 
$(u, v)$ and $(v, u)$ if and only if the edge $uv$ of $G$ is not 
on a chordless cycle of length 4. 
We remove such components from $G^+$, and 
let $c_1, c_2, \ldots, c_k$ be the remaining components. 
To each component $c_i$, we assign the Boolean variable $x_i$. 
Since we assume that $G^+$ is bipartite, 
the vertices of each component can be partitioned 
into two color classes. 
We assign the literal $x_i$ to the vertices of one color class of $c_i$, and 
we assign the literal $\overline{x_i}$ (the negation of $x_i$) 
to the vertices of the other color class of $c_i$. 
\par
Notice that a vertex $(u, v)$ of $G^+$ has its literal 
if and only if the edge $uv$ of $G$ is on a cycle of length 4. 
Let $l_{(u, v)}$ be the literal assigned to a vertex $(u, v)$ of $G^+$. 
The 2CNF formula $\phi$ consists of all the clauses 
$(l_{(u, v)} \vee l_{(v, w)})$ such that 
the vertices $(u, v, w)$ in $G$ form 
a path of three vertices with $(w, u) \in \overline{F}$. 
When no literal is assigned to $(u, v)$ or $(v, w)$, 
the 2CNF formula $\phi$ does not contain the clause. 
\par
Let $\tau$ be a truth assignment of the variables in $\phi$. 
Notice that $l_{(u, v)} = 0$ in $\tau$ if and only if $l_{(v, u)} = 1$ 
since each vertex $(u, v)$ of $G^+$ is adjacent to $(v, u)$. 
We obtain the partial orientation $F$ of $G$ from $\tau$ 
by orienting each edge $uv \in E(G)$ as $(u, v) \in F$ 
if $l_{(u, v)} = 0$ in $\tau$. 
It is obvious from the construction of $\phi$ that 
a truth assignment $\tau$ satisfies $\phi$ if and only if 
$F \cup \overline{F}$ contains no $\Delta$-obstructions. 
We can also see that 
the edges on every chordless cycle of length 4 is alternately oriented 
and the other edges remain undirected. 
\par
The 2CNF formula $\phi$ has at most $2m$ Boolean variables 
and at most $nm$ clauses. 
We can also see that $\phi$ can be constructed in $O(nm)$ time, 
since all the paths of three vertices can be found in $O(nm)$ time. 
Since a satisfying truth assignment of a 2CNF formula $\phi$ can be computed 
in time linear to the size of $\phi$ (see~\cite{APT79-IPL} for example), 
\textbf{Step~2} can be performed in $O(nm)$ time. 

It is easy to see that \textbf{Step~3} and \textbf{Step~4} 
can be performed in $O(nm)$ time and $O(n^2)$ time, respectively. 
Now, we have the following. 
\begin{theorem}
Simple-triangle graphs can be recognized in $O(nm)$ time, and 
an apex ordering of a simple-triangle graph 
can be computed in the same time bound. 
\end{theorem}

From \textbf{Step~1} and \textbf{Step~2} of the algorithm, 
we also have the following as a byproduct. 
\begin{theorem}\label{theorem:alternately-orientable-cocomparability}
Alternately orientable cocomparability graphs can be recognized in $O(nm)$ time, and 
alternating orientation of a cocomparability graph 
can be computed in the same time bound. 
\end{theorem}

\section{Correctness of the algorithm}\label{section:correctness}
In this section, we prove correctness of our algorithm. 
In order to prove the correctness, it suffices to show that 
when \textbf{Step~3} is finished, $F'' \cup \overline{F}$ is acyclic, 
where $F''$ is the resultant orientation of \textbf{Step~3}. 

Recall that the graph $G$ is an alternately orientable graph and 
the complement $\overline{G}$ of $G$ is a comparability graph 
with a transitive orientation $\overline{F}$. 
Since $G$ is a cocomparability graph, 
the length of every chordless cycle in $G$ is at most 4. 
Recall also that $G$ has a partial orientation $F$ that 
satisfies the following three conditions: 
(1) $F \cup \overline{F}$ contains no $\Delta$-obstructions, 
(2) the edges on every chordless cycle of length 4 are oriented so that 
the directions of the edges alternate, and 
(3) the other edges remain undirected. 
For a vertex $v$ of $G$, 
let $F_v$ be the set of all the edges $(w, u) \in F$ 
such that $(u, v, w)$ is a directed cycle in $F$, and 
let $F'$ be the partial orientation of $G$ obtained from $F$ 
by reversing the orientation of all the edges in $F_v$, that is, 
$F' = (F - F_v) \cup F_v^{-1}$. 

The outline of the proof is as follows. 
An \emph{alternating $2k$-cycle} in $F \cup \overline{F}$ 
with $k \geq 2$ is a directed cycle 
$(a_0, b_0, a_1, b_1, \ldots, a_{k-1}, b_{k-1})$ of length $2k$ 
with $(a_i, b_i) \in F$ and $(b_i, a_{i+1}) \in \overline{F}$ 
for any $i = 0, 1, \ldots, k-1$ (indices are modulo $k$). 
See Figure~\ref{figs} for example. 
We first show that
\begin{itemize}
\item $F$ is acyclic if and only if 
$F \cup \overline{F}$ contains no alternating 6-cycles.
\end{itemize}
We next show the following: 
\begin{itemize}
\item $F' \cup \overline{F}$ contains no alternating 6-cycles having the vertex $v$, 
\item the reversing the direction of the edges in $F_v$ generates no alternating 6-cycles, 
\item the directions of the edges still alternate in $F'$ 
on every chordless cycle of length 4, and 
\item $F' \cup \overline{F}$ is still contains no $\Delta$-obstructions. 
\end{itemize}
Thus continuing in this way for each vertex, 
we obtain the acyclic orientation $F''$ such that 
$F'' \cup \overline{F}$ contains no $\Delta$-obstructions, 
that is, $F'' \cup \overline{F}$ is acyclic. 

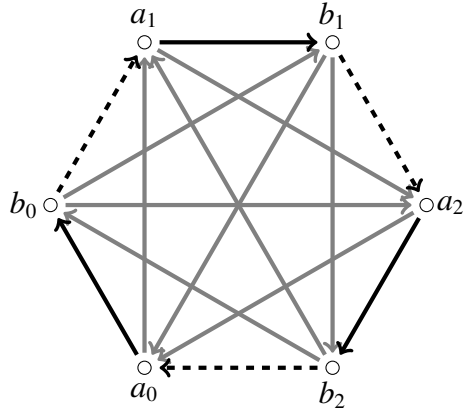
\begin{figure*}[t]
\centering
\begin{tikzpicture}[shorten >=3pt,shorten <=3pt,->]
\tikzstyle{every node}=[draw,circle,fill=white,minimum size=5pt,
                        inner sep=0pt]
\node [label=right:$a_2$] (a2) at (0:2.5) {};
\node [label=above:$b_1$] (b1) at (60:2.5) {};
\node [label=above:$a_1$] (a1) at (120:2.5) {};
\node [label=left:$b_0$] (b0) at (180:2.5) {};
\node [label=below:$a_0$] (a0) at (240:2.5) {};
\node [label=below:$b_2$] (b2) at (300:2.5) {};
\draw [ultra thick,color=gray] (a0) -- (a1);
\draw [ultra thick,color=gray] (a1) -- (a2);
\draw [ultra thick,color=gray] (a2) -- (a0);
\draw [ultra thick,color=gray] (b0) -- (b1);
\draw [ultra thick,color=gray] (b1) -- (b2);
\draw [ultra thick,color=gray] (b2) -- (b0);
\draw [ultra thick,color=gray] (b0) -- (a2);
\draw [ultra thick,color=gray] (b1) -- (a0);
\draw [ultra thick,color=gray] (b2) -- (a1);
\draw [ultra thick] (a0) -- (b0);
\draw [ultra thick] (a1) -- (b1);
\draw [ultra thick] (a2) -- (b2);
\draw [ultra thick,dashed] (b0) -- (a1);
\draw [ultra thick,dashed] (b1) -- (a2);
\draw [ultra thick,dashed] (b2) -- (a0);
\end{tikzpicture}
  \caption{
    An alternating 6-cycles. 
    An arrow $a \to b$ denotes edge $(a, b) \in F$, and 
    a dashed arrow $a \dasharrow b$ denotes edge $(a, b) \in \overline{F}$. 
    }
  \label{figs}
\end{figure*}

\subsection{Definitions and facts}
We begin to state some definitions and facts.
\begin{fact}\label{fact:diagonal}
Let $a_0, b_0, a_1, b_1$ be four vertices of $G$. 
If $a_0b_0, a_1b_1 \in E(G)$ and $(a_0, b_1), (a_1, b_0) \in \overline{F}$, 
then $a_0a_1, b_0b_1 \in E(G)$, that is, 
$(a_0, a_1, b_1, b_0)$ is a chordless cycle of length 4. 
\end{fact}
\begin{proof}
Recall that $\overline{F}$ is a transitive orientation. 
If $(a_0, a_1) \in \overline{F}$, then 
$(a_1, b_0) \in \overline{F}$ implies $(a_0, b_0) \in \overline{F}$, 
contradicting $a_0b_0 \in E(G)$. 
If $(a_1, a_0) \in \overline{F}$, then 
$(a_0, b_1) \in \overline{F}$ implies $(a_1, b_1) \in \overline{F}$, 
contradicting $a_1b_1 \in E(G)$. 
Thus $a_0a_1 \in E(G)$. 
Similarly, if $(b_0, b_1) \in \overline{F}$, then 
$(a_1, b_0) \in \overline{F}$ implies $(a_1, b_1) \in \overline{F}$, 
contradicting $a_1b_1 \in E(G)$. 
If $(b_1, b_0) \in \overline{F}$, then 
$(a_0, b_1) \in \overline{F}$ implies $(a_0, b_0) \in \overline{F}$, 
contradicting $a_0b_0 \in E(G)$. 
Thus $b_0b_1 \in E(G)$. 
\end{proof}

Suppose that $F \cup \overline{F}$ contains an alternating 4-cycle 
$(a_0, b_0, a_1, b_1)$ with $(a_0, b_0), (a_1, b_1) \in F$ and 
$(b_0, a_1), (b_1, a_0) \in \overline{F}$. 
We have from Claim~\ref{fact:diagonal} that 
$(a_0, a_1, b_1, b_0)$ is a chordless cycle of length 4. 
The directions of the edges on this cycle must alternate in $F$, 
but $(a_0, b_0), (a_1, b_1) \in F$, a contradiction. 
Thus $F \cup \overline{F}$ contains no alternating 4-cycles. 
An \emph{alternating 4-anticycle} of $F \cup \overline{F}$ is 
a subgraph consisting of four vertices 
$a_0, b_0, a_1, b_1$ with $(a_0, b_0), (a_1, b_1) \in F$ and 
$(a_0, b_1), (a_1, b_0) \in \overline{F}$. 
We have from Claim~\ref{fact:diagonal} that 
$(a_0, a_1, b_1, b_0)$ is a chordless cycle of length 4. 
The directions of the edges on this cycle must alternate in $F$, 
but $(a_0, b_0), (a_1, b_1) \in F$, a contradiction. 
Thus $F \cup \overline{F}$ contains no alternating 4-anticycles. 

The following claim states the direction of 
chords of an alternating 6-cycle. 
\begin{fact}\label{fact:AC6}
If $F \cup \overline{F}$ contains 
an alternating 6-cycle $(a_0, b_0, a_1, b_1, a_2, b_2)$ 
with $(a_i, b_i) \in F$ and $(b_i, a_{i+1}) \in \overline{F}$ 
for any $i = 0, 1, 2$ (indices are modulo 3), 
then 
$(b_2, a_1), (b_2, b_0), (a_0, a_1), 
 (b_0, a_2), (b_0, b_1), (a_1, a_2), 
 (b_1, a_0), (b_1, b_2), (a_2, a_0) \in F$. 
See Figure~\ref{figs}. 
\end{fact}
\begin{proof}
If $(b_2, a_1) \in \overline{F}$, then $(a_1, b_1, a_2, b_2)$ is 
an alternating 4-cycle, a contradiction. 
If $(a_1, b_2) \in \overline{F}$, then 
$(b_0, a_1), (b_2, a_0) \in \overline{F}$ implies $(b_0, a_0) \in \overline{F}$, 
contradicting $(a_0, b_0) \in F$. 
Thus $b_2a_1 \in E(G)$. 
We have from Claim~\ref{fact:diagonal} that 
$(a_0, b_0, b_2, a_1)$ is a chordless cycle of length 4. 
Since the directions of the edges on this cycle alternate in $F$, 
we have $(b_2, a_1), (b_2, b_0), (a_0, a_1) \in F$. 
By similar arguments, we also have 
$(b_0, a_2), (b_0, b_1), (a_1, a_2) \in F$ and 
$(b_1, a_0), (b_1, b_2), (a_2, a_0) \in F$. 
\end{proof}

Recall that for any directed edge $(u, v) \in F$, 
there are another two vertices $w, z$ of $G$ such that 
$(u, v, w, z)$ is a chordless cycle of length 4, 
since an edge of $G$ remains undirected if it is not 
on a chordless cycle of length 4. 
\begin{fact}\label{fact:A}
If there are four vertices $a, b, c, d$ with 
$(a, b), (b, c), (c, d) \in F$ and $(d, a) \in \overline{F}$, 
then there are another two vertices $e, f$ with 
$(b, e), (f, c), (f, e) \in F$ and $bf, ce \in E(\overline{G})$. 
In addition, 
\begin{itemize}
\item $(a, f), (d, b), (d, f), (c, a), (e, a), (e, d) \in F$ and 
\item there is an alternating 6-cycle 
consisting of the vertices $a, b, c, d, e, f$. 
\end{itemize}
\end{fact}
\begin{proof}
Since $(b, c)$ is oriented in $F$, 
there are another two vertices $e, f$ with 
$(b, e), (f, c), (f, e) \in F$ and 
$bf, ce \in E(\overline{G})$. 
If $(a, c) \in \overline{F}$, then 
$(d, a) \in \overline{F}$ implies $(d, c) \in \overline{F}$, 
contradicting $(c, d) \in F$. 
If $(c, a) \in \overline{F}$, then $(a, b, c)$ 
is a $\Delta$-obstruction, a contradiction. 
Thus $ac \in E(G)$. 
Similarly, 
if $(b, d) \in \overline{F}$, then 
$(d, a) \in \overline{F}$ implies $(b, a) \in \overline{F}$, 
contradicting $(a, b) \in F$. 
If $(d, b) \in \overline{F}$, then $(b, c, d)$ 
is a $\Delta$-obstruction, a contradiction. 
Thus $bd \in E(G)$. 
\par
If $(e, a) \in \overline{F}$, then $(a, b, e)$ 
is a $\Delta$-obstruction, a contradiction. 
Suppose $(a, e) \in \overline{F}$. 
If $af \in E(G)$, then $(a, b, e, f)$ is a chordless cycle of length 4 
with $(a, b), (b, e) \in F$, 
contradicting the definition of $F$. 
Thus $af \in E(\overline{G})$. 
If $(f, a) \in \overline{F}$, then 
$(a, e) \in \overline{F}$ implies $(f, e) \in \overline{F}$, 
contradicting $(f, e) \in F$. 
If $(a, f) \in \overline{F}$, then 
$(d, a) \in \overline{F}$ implies $(d, f) \in \overline{F}$, 
but then $(f, c, d)$ is a $\Delta$-obstruction, a contradiction. 
Thus $ae \in E(G)$. 
\par
If $af \in E(\overline{G})$, then $(a, e, f, c)$ 
is a chordless cycle of length 4. 
Since $(f, c), (f, e) \in F$, 
we have $(a, c), (a, e) \in F$, but then $(a, c, d)$ 
is a $\Delta$-obstruction, a contradiction. 
Thus $af \in E(G)$. 
If $(d, f) \in \overline{F}$, then $(f, c, d)$ 
is a $\Delta$-obstruction, a contradiction. 
If $(f, d) \in \overline{F}$, then 
$(d, a) \in \overline{F}$ implies $(f, a) \in \overline{F}$, 
contradicting $af \in E(G)$. 
Thus $df \in E(G)$. 
If $de \in E(\overline{G})$, then $(d, b, e, f)$ 
is a chordless cycle of length 4. 
Since $(b, e), (f, e) \in F$, 
we have $(b, d), (f, d) \in F$, but then $(a, b, d)$ 
is a $\Delta$-obstruction, a contradiction. 
Thus $de \in E(G)$. 
\par
Since $(a, b, d, f)$ is a chordless cycle of length 4 
and $(a, b) \in F$, we have $(a, f), (d, b), (d, f) \in F$. 
Since $(a, e, d, c)$ is a chordless cycle of length 4 
and $(c, d) \in F$, we also have $(c, a), (e, a), (e, d) \in F$. 
We can verify that 
for any directions are assigned to the edges $bf, ce \in E(\overline{G})$, 
there is an alternating 6-cycle 
consisting of the vertices $a, b, c, d, e, f$. 
\end{proof}

\begin{fact}\label{fact:C3-AC6}
If there is a directed cycle $(a, b, c)$ in $F$, then 
there is an alternating 6-cycle of $F \cup \overline{F}$ 
containing the vertices $a, b, c$. 
\end{fact}
\begin{proof}
Since $(a, b)$ is oriented in $F$, 
there are another two vertices $d, e$ with 
$(a, d), (e, b), (e, d) \in F$ and 
$ae, bd \in E(\overline{G})$. 
We have $c \neq d$ from $bc \in E(G)$ and $bd \in E(\overline{G})$. 
We also have $c \neq e$ from $(b, c), (e, b) \in F$. 
If $(a, e) \in \overline{F}$, then we have the claim from Claim~\ref{fact:A} 
since $(e, b, c, a)$ is a directed cycle of length 4 in $F \cup \overline{F}$. 
If $(d, b) \in \overline{F}$, then we also have the claim from Claim~\ref{fact:A} 
since $(b, c, a, d)$ is a directed cycle of length 4 in $F \cup \overline{F}$. 
Therefore, we assume $(e, a), (b, d) \in \overline{F}$. 
\par
Since $(b, c)$ is oriented in $F$, 
there are another two vertices $f, g$ with 
$(b, f), (g, c), (g, f) \in F$ and 
$bg, cf \in E(\overline{G})$. 
We have $a \neq f$ from $(a, b), (b, f) \in F$. 
We also have $a \neq g$ from $ab \in E(G)$ and $bg \in E(\overline{G})$. 
If $(b, g) \in \overline{F}$, then we have the claim from Claim~\ref{fact:A} 
since $(g, c, a, b)$ is a directed cycle of length 4 in $F \cup \overline{F}$. 
If $(f, c) \in \overline{F}$, then we also have the claim from Claim~\ref{fact:A} 
since $(c, a, b, f)$ is a directed cycle of length 4 in $F \cup \overline{F}$. 
Therefore, we assume $(g, b), (c, f) \in \overline{F}$. 
\par
We have $d \neq f$ from $bd \in E(\overline{G})$ and $bf \in E(G)$. 
We also have $d \neq g$ from $(b, d), (g, b) \in \overline{F}$. 
Similarly, 
we have $e \neq f$ from $(e, b), (b, f) \in F$. 
We also have $e \neq g$ from $be \in E(G)$ and $bg \in E(\overline{G})$. 
Therefore, the seven vertices $a, b, c, d, e, f, g$ are distinct. 
\par
If $(e, f) \in \overline{F}$, then four vertices $b, f, e, d$ form 
an alternating 4-anticycle, a contradiction. 
If $(f, e) \in \overline{F}$, then $(e, b, f)$ is 
a $\Delta$-obstruction, a contradiction. 
Thus $ef \in E(G)$. 
Now, we have from Claim~\ref{fact:diagonal} that 
$(a, c, e, f)$ is a chordless cycle of length 4. 
Since $(c, a) \in F$, we have $(c, e), (f, a), (f, e) \in F$. 
\par
If $(d, c) \in \overline{F}$, then 
$(b, d) \in \overline{F}$ implies $(b, c) \in \overline{F}$, 
contradicting $(b, c) \in F$. 
If $(c, d) \in \overline{F}$, then four vertices $c, a, e, d$ form 
an alternating 4-anticycle, a contradiction. 
Thus $cd \in E(G)$. 
Now, we have from Claim~\ref{fact:diagonal} that 
$(c, d, f, b)$ is a chordless cycle of length 4. 
Since $(b, c), (b, f) \in F$, we have $(d, c), (d, f) \in F$. 
\par
Now, we can see that $(a, b, d, c, f, e)$ is an alternating 6-cycle 
in $F \cup \overline{F}$. 
\end{proof}

\begin{fact}\label{fact:C3-diamond}
If there is a directed cycle $(a, b, c)$ in $F$, then 
there is another vertex $d$ with 
$(c, d), (d, b) \in F$ and either 
$(a, d) \in \overline{F}$ or $(d, a) \in \overline{F}$. 
\end{fact}
\begin{proof}
We have the claim from Claims~\ref{fact:C3-AC6} and~\ref{fact:AC6}. 
\end{proof}

\subsection{Main part of the proof}
Now, we show the following series of claims, 
which proves the correctness of the algorithm. 
It is immediate from Claim~\ref{fact:AC6} that 
if $F \cup \overline{F}$ contains an alternating 6-cycle, 
then $F$ contains a directed cycle of length 3. 
The following claim states that the converse is also true. 
(Recall that $F$ contains undirected edges. 
Thus $F$ may contain directed cycles of length greater than 3, 
and Claim~\ref{fact:C3-AC6} does not imply the converse. )
\begin{claim}
If $F$ is not acyclic, then 
$F \cup \overline{F}$ contains an alternating 6-cycle.
\end{claim}
\begin{proof}
Before proving the claim, 
we introduce the \emph{middle edge} of a chordless cycle of length 4. 
Let $C_4 = (a, b, c, d)$ be a cycle of length 4 in $G$ with 
$(a, b), (a, d), (c, b), (c, d) \in F$ and 
$(a, c), (b, d) \in \overline{F}$. 
We call the edge $(c, b)$ the \emph{middle edge} of $C_4$ in $F \cup \overline{F}$. 
We also say an edge is a middle edge in $F \cup \overline{F}$ 
if there is a chordless cycle of length 4 such that 
the edge is the middle edge of the cycle. 
Notice that for any edge $(u, v)$ on $C_4$, 
there is the directed path from $u$ to $v$ in $F \cup \overline{F}$ 
consisting of the edges in $\overline{F}$ and the middle edge. 
For example, there is the directed path $(a, c, b, d)$ 
for the edge $(a, d)$. 
\par
Recall that each directed edge in $F$ is on a cycle of length 4. 
Thus if $F$ contains a directed cycle, then 
$F \cup \overline{F}$ contains 
a directed cycle $C = (v_0, v_1, \ldots, v_{k-1})$ with $k \geq 3$ 
consisting of the edges in $\overline{F}$ and 
the middle edges in $F \cup \overline{F}$. 
\par
We now show that 
if $(v_i, v_{i+1})$ and $(v_{i+1}, v_{i+2})$ are middle edges 
in $F \cup \overline{F}$ for some $i$ (indices are modulo $k$), 
then either $(v_i, v_{i+2}) \in F$ or $(v_{i+2}, v_i) \in F$. 
In addition, we show that if $(v_i, v_{i+2}) \in F$, then 
$(v_i, v_{i+2})$ is a middle edge in $F \cup \overline{F}$. 
Since $(v_i, v_{i+1})$ is a middle edge, 
there are another two vertices $a, b$ with 
$(v_i, a), (b, v_{i+1}), (b, a) \in F$ and 
$(b, v_i), (v_{i+1}, a) \in \overline{F}$. 
Similarly, 
since $(v_{i+1}, v_{i+2})$ is a middle edge, 
there are another two vertices $c, d$ with 
$(v_{i+1}, c), (d, v_{i+2}), (d, c) \in F$ and 
$(v_{i+2}, c), (d, v_{i+1}) \in \overline{F}$. 
We have $v_i \neq c$ from $(v_i, v_{i+1}), (v_{i+1}, c) \in F$. 
We also have $v_i \neq d$ from $v_iv_{i+1} \in E(G)$ and $v_{i+1}d \in E(\overline{G})$. 
Similarly, 
we have $a \neq v_{i+2}$ from $av_{i+1} \in E(\overline{G})$ and $v_{i+1}v_{i+2} \in E(G)$. 
We also have $a \neq c$ 
from $av_{i+1} \in E(\overline{G})$ and $v_{i+1}c \in E(G)$. 
Moreover, we have $a \neq d$ from $(v_{i+1}, a), (d, v_{i+1}) \in \overline{F}$. 
By similar arguments, we have $b \neq v_{i+2}, c, d$. 
Therefore, the seven vertices $v_i, v_{i+1}, v_{i+2}, a, b, c, d$ are distinct. 
\par
If $(b, c) \in \overline{F}$, then four vertices $v_{i+1}, c, b, a$ form 
an alternating 4-anticycle, a contradiction. 
If $(c, b) \in \overline{F}$, then $(b, v_{i+1}, c)$ is 
a $\Delta$-obstruction, a contradiction. 
Thus $bc \in E(G)$. 
If $(v_i, v_{i+2}) \in \overline{F}$, then 
$(b, v_i), (v_{i+2}, c) \in \overline{F}$ implies $(b, c) \in \overline{F}$, 
contradicting $bc \in E(G)$. 
If $(v_{i+2}, v_i) \in \overline{F}$, then $(v_i, v_{i+1}, v_{i+2})$ is 
a $\Delta$-obstruction, a contradiction. 
Thus $v_iv_{i+2} \in E(G)$. 
Now, we have from Claim~\ref{fact:diagonal} that 
$(v_i, v_{i+2}, b, c)$ is a chordless cycle of length 4, 
and hence the edge $v_iv_{i+2}$ is oriented in $F$. 
When $(v_i, v_{i+2}) \in F$, 
the edge $(v_i, v_{i+2})$ is the middle edge 
of the cycle $(v_i, v_{i+2}, b, c)$. 
\par
We have from Claim~\ref{fact:C3-AC6} that 
if $(v_i, v_{i+1}), (v_{i+1}, v_{i+2}), (v_{i+2}, v_i) \in F$, then 
$F \cup \overline{F}$ contains an alternating 6-cycle, and 
we have the claim. 
Thus we can assume that the directed cycle $C$ contains 
no two middle edges that are consecutive on $C$. 
Obviously, if $(v_i, v_{i+1}), (v_{i+1}, v_{i+2}) \in \overline{F}$ 
for some $i$ (indices are modulo $k$), then 
$(v_i, v_{i+2}) \in \overline{F}$. 
Therefore, we can assume that $C$ contains 
no two edges in $\overline{F}$ that are consecutive on $C$, 
that is, $C$ is an alternating cycle. 
\par
It remains to show that if 
$F \cup \overline{F}$ contains an alternating $2k$-cycle with $k \geq 4$, then 
$F \cup \overline{F}$ also contains an alternating 6-cycle. 
Suppose that $F \cup \overline{F}$ contains an alternating $2k$-cycle 
with $k \geq 4$ but it contains no alternating 6-cycles. 
Let $(a_0, b_0, a_1, b_1, \ldots, a_{k-1}, b_{k-1})$ be 
such an alternating $2k$-cycle with 
$(a_i, b_i) \in F$ and $(b_i, a_{i+1}) \in \overline{F}$ 
for any $i = 0, 1, \ldots, k-1$ (indices are modulo $k$). 
We assume without loss of generality that 
the length of this cycle is minimal, that is, 
$F \cup \overline{F}$ contains no alternating cycle 
of length smaller than $2k$. 
Since $k \geq 4$, the vertices 
$a_0, b_0, a_1, b_1, a_2, b_2, a_3, b_3$ are distinct. 
If $(b_2, a_1) \in \overline{F}$, then 
$(a_1, b_1, a_2, b_2)$ is an alternating 4-cycle, a contradiction. 
If $(b_0, a_3) \in \overline{F}$, then 
$(a_0, b_0, a_3, b_3, \ldots, a_{k-1}, b_{k-1})$ 
is an alternating $2(k-2)$-cycle, a contradiction. 
If $(a_1, b_2) \in \overline{F}$, then 
$(b_0, a_1), (b_2, a_3) \in \overline{F}$ implies $(b_0, a_3) \in \overline{F}$, 
a contradiction. 
If $(a_3, b_0) \in \overline{F}$, then 
$(b_2, a_3), (b_0, a_1) \in \overline{F}$ implies $(b_2, a_1) \in \overline{F}$, 
a contradiction. 
Thus $a_1b_2, b_0a_3 \in E(G)$. 
Now, we have from Claim~\ref{fact:diagonal} that 
$(a_1, b_2, b_0, a_3)$ is a chordless cycle of length 4. 
If $(a_1, b_2), (a_1, a_3), (b_0, b_2), (b_0, a_3) \in F$, then 
$(a_0, b_0, a_1, b_2, a_3, b_3, \ldots, a_{k-1}, b_{k-1})$ 
is an alternating $2(k-1)$-cycle, a contradiction. 
If $(b_2, a_1), (a_3, a_1), (b_2, b_0), (a_3, b_0) \in F$, then 
$(a_1, b_1, a_2, b_2, a_3, b_0)$ is an alternating 6-cycle, a contradiction. 
\end{proof}

Recall that 
$F_v$ is the set of all the edges $(w, u) \in F$ 
such that $(u, v, w)$ is a directed cycle in $F$. 
Recall also that $F'$ is the partial orientation of $G$ obtained from $F$ 
by reversing the orientation of all the edges in $F_v$, that is, 
$F' = (F - F_v) + F_v^{-1}$. 
The following claim states that 
$F' \cup \overline{F}$ contains no alternating 6-cycles having the vertex $v$. 
\begin{claim}\label{claim:2}
There is no alternating 6-cycles of $F' \cup \overline{F}$ 
having the vertex $v$. 
\end{claim}
\begin{proof}
We prove the claim by contradiction. 
Suppose that $F' \cup \overline{F}$ contains 
an alternating 6-cycle $(a_0, b_0, a_1, b_1, a_2, b_2)$ with 
$(a_i, b_i) \in F'$ and $(b_i, a_{i+1}) \in \overline{F}$ 
for any $i = 0, 1, 2$ (indices are modulo 3). 
Suppose $v = a_0$. 
It is obvious that $(a_0, b_0), (a_2, b_2) \notin F_v^{-1}$, 
that is, $(a_0, b_0), (a_2, b_2) \in F$. 
If $(a_1, b_1) \notin F_v^{-1}$, then 
we have from Claim~\ref{fact:AC6} that 
$(a_0, a_1), (a_1, b_1), (b_1, a_0) \in F$, and then 
$(a_1, b_1) \in F_v$, a contradiction. 
If $(a_1, b_1) \in F_v^{-1}$, then 
$(a_0, b_1), (b_1, a_1), (a_1, a_0) \in F$, 
but then $F \cup \overline{F}$ contains 
an alternating 4-cycle $(a_0, b_1, a_2, b_2)$, a contradiction. 
Thus $v \neq a_0$. 
By similar arguments, we have $v \neq a_1, a_2$. 
\par
Suppose $v = b_0$. 
It is obvious that $(a_0, b_0), (a_1, b_1) \notin F_v^{-1}$, 
that is, $(a_0, b_0), (a_1, b_1) \in F$. 
If $(a_2, b_2) \notin F_v^{-1}$, then 
we have from Claim~\ref{fact:AC6} that 
$(b_0, a_2), (a_2, b_2), (b_2, b_0) \in F$, and then 
$(a_2, b_2) \in F_v$, a contradiction. 
If $(a_2, b_2) \in F_v^{-1}$, then 
$(b_0, b_2), (b_2, a_2), (a_2, b_0) \in F$, 
but then $F \cup \overline{F}$ contains 
an alternating 4-cycle $(a_2, b_0, a_1, b_1)$, a contradiction. 
Thus $v \neq b_0$. 
By similar arguments, we have $v \neq b_1, b_2$. 
\end{proof}

The following claim states that 
the reversing the orientation of the edges 
in $F_v$ generates no alternating 6-cycles. 
\begin{claim}
No edges in $F_v^{-1}$ is an edge of any alternating 6-cycle 
of $F' \cup \overline{F}$. 
\end{claim}
\begin{proof}
We prove the claim by contradiction. 
Suppose that $F' \cup \overline{F}$ contains 
an alternating 6-cycle $(a_0, b_0, a_1, b_1, a_2, b_2)$ with 
$(a_i, b_i) \in F'$ and $(b_i, a_{i+1}) \in \overline{F}$ 
for any $i = 0, 1, 2$ (indices are modulo 3). 
We have from Claim~\ref{claim:2} that 
the vertex $v$ is not on the alternating 6-cycle. 
If $(a_0, b_0), (a_1, b_1) \in F_v^{-1}$, then 
$(v, b_0), (b_0, a_0), (a_0, v) \in F$ and 
$(v, b_1), (b_1, a_1), (a_1, v) \in F$, 
but then $(a_1, v, b_0)$ is a $\Delta$-obstruction in $F \cup \overline{F}$, 
a contradiction. 
Thus we can see that 
at most one edge on the alternating 6-cycle is in $F_v^{-1}$. 
\par
We assume without loss of generality that $(a_0, b_0) \in F_v^{-1}$ 
and $(a_1, b_1), (a_2, b_2) \notin F_v^{-1}$. 
We have $(v, b_0), (b_0, a_0), (a_0, v) \in F$. 
It is obvious from $a_1b_0, a_0b_2 \in E(\overline{G})$ that $v \neq a_1, b_2$. 
We have $v \neq b_1$ since otherwise 
$(a_1, b_1 = v, b_0)$ is a $\Delta$-obstruction, a contradiction. 
We also have $v \neq a_2$ since otherwise 
$(a_0, v = a_2, b_2)$ is a $\Delta$-obstruction, a contradiction. 
Therefore, the seven vertices $v, a_0, b_0, a_1, b_1, a_2, b_2$ are distinct. 
\par
We have from Claim~\ref{fact:C3-diamond} that 
there is a vertex $z$ such that $(a_0, z), (z, b_0) \in F$ and 
either $(v, z) \in \overline{F}$ or $(z, v) \in \overline{F}$. 
We first suppose $(z, v) \in \overline{F}$. 
It is obvious from $a_1b_0, a_0b_2 \in E(\overline{G})$ that $z \neq a_1, b_2$. 
We have $z \neq b_1$ since otherwise 
$(a_1, b_1 = z, b_0)$ is a $\Delta$-obstruction, a contradiction. 
We also have $z \neq a_2$ since otherwise 
$(a_0, z = a_2, b_2)$ is a $\Delta$-obstruction, a contradiction. 
Therefore, the eight vertices $v, z, a_0, b_0, a_1, b_1, a_2, b_2$ are distinct. 
\par
If $(b_1, v) \in \overline{F}$, then 
$(v, b_0, a_1, b_1)$ is an alternating 4-cycle, a contradiction. 
If $(z, a_2) \in \overline{F}$, then 
$(a_2, b_2, a_0, z)$ is an alternating 4-cycle, a contradiction. 
If $(v, b_1) \in \overline{F}$, then 
$(z, v), (b_1, a_2) \in \overline{F}$ implies $(z, a_2) \in \overline{F}$, 
a contradiction. 
If $(a_2, z) \in \overline{F}$, then 
$(b_1, a_2), (z, v) \in \overline{F}$ implies $(b_1, v) \in \overline{F}$, 
a contradiction. 
Thus $vb_1, za_2 \in E(G)$. 
Now, we have from Claim~\ref{fact:diagonal} that 
$(v, b_1, z, a_2)$ is a chordless cycle of length 4. 
\par
If $(b_1, v), (b_1, z), (a_2, v), (a_2, z) \in F$, then 
$(v, b_0, a_1, b_1, a_2, z)$ is an alternating 6-cycle 
in $F \cup \overline{F}$, and hence we have from Claim~\ref{fact:AC6} 
that $(a_1, b_1) \in F_v$, a contradiction. 
If $(v, b_1), (v, a_2), (z, b_1), (z, a_2) \in F$, then 
$(v, b_1, a_2, b_2, a_0, z)$ is an alternating 6-cycle 
in $F \cup \overline{F}$, and hence we have from Claim~\ref{fact:AC6} 
that $(a_2, b_2) \in F_v$, a contradiction. 
\par
When we suppose $(v, z) \in \overline{F}$, 
we also have either $(a_1, b_1) \in F_v$ or $(a_2, b_2) \in F_v$, 
a contradiction. Thus we have the claim. 
\end{proof}

The following claim states that 
the directions of the edges still alternate in $F'$ 
on every chordless cycle of length 4. 
\begin{claim}
If an edge on a chordless cycle  of length 4 is in $F_v$, 
then all the edges on the cycle is in $F_v$. 
\end{claim}
\begin{proof}
Suppose that there is a chordless cycle $(a, b, c, d)$ of length 4 in $G$ 
such that $(a, b) \in F_v$, that is, 
$(v, a), (a, b), (b, v), (a, d), (c, b), (c, d) \in F$ 
and $ac, bd \in E(\overline{G})$. 
We have from Claim~\ref{fact:C3-diamond} that 
there is a vertex $z$ such that $(b, z), (z, a) \in F$ and 
either $(v, z) \in \overline{F}$ or $(z, v) \in \overline{F}$. 
In either case, we have from Claim~\ref{fact:A} that 
$(v, c), (d, v) \in F$. 
Thus $(a, d), (c, b), (c, d) \in F_v$. 
\end{proof}

The following claim states that 
the reversing the orientation of the edges 
in $F_v$ generates no $\Delta$-obstructions, 
and thus $F' \cup \overline{F}$ still contains 
no $\Delta$-obstructions. 
\begin{claim}
No edges in $F_v^{-1}$ is an edge of any $\Delta$-obstruction 
of $F' \cup \overline{F}$. 
\end{claim}
\begin{proof}
We prove the claim by contradiction. 
Suppose that $F' \cup \overline{F}$ contains a $\Delta$-obstruction 
$(a, b, c)$ with $(a, b), (b, c) \in F'$ and $(c, a) \in \overline{F}$. 
If $(a, b), (b, c) \in F_v^{-1}$, then 
$(v, b), (b, a), (a, v) \in F$ and $(v, c), (c, b), (b, v) \in F$, a contradiction. 
\par
Suppose $(a, b) \in F_v^{-1}$ and $(b, c) \notin F_v^{-1}$. 
We have $(v, b), (b, a), (a, v), (b, c) \in F$ and $(c, a) \in \overline{F}$. 
We have from Claim~\ref{fact:C3-diamond} that 
there is a vertex $z$ such that $(a, z), (z, b) \in F$ and 
either $(v, z) \in \overline{F}$ or $(z, v) \in \overline{F}$. 
We first suppose $(v, z) \in \overline{F}$. 
If $(z, c) \in \overline{F}$, then 
$(c, a) \in \overline{F}$ implies $(z, a) \in \overline{F}$, 
contradicting $(a, z) \in F$. 
If $(c, z) \in \overline{F}$, then $(z, b, c)$ is 
a $\Delta$-obstruction, a contradiction. 
Thus $zc \in E(G)$. 
We now have from Claim~\ref{fact:diagonal} that 
$(z, c, v, a)$ is a chordless cycle of length 4. 
Since $(a, v), (a, z) \in F$, we have $(c, v), (c, z) \in F$, 
but then $(v, b), (b, c), (c, v) \in F$ implies $(b, c) \in F_v$, 
a contradiction. 
We next suppose $(z, v) \in \overline{F}$. 
If $(v, c) \in \overline{F}$, then 
$(c, a) \in \overline{F}$ implies $(v, a) \in \overline{F}$, 
contradicting $(a, v) \in F$. 
If $(c, v) \in \overline{F}$, then $(v, b, c)$ is 
a $\Delta$-obstruction, a contradiction. 
Thus $vc \in E(G)$. 
We now have from Claim~\ref{fact:diagonal} that 
$(v, c, z, a)$ is a chordless cycle of length 4. 
Since $(a, v), (a, z) \in F$, we have $(c, v), (c, z) \in F$, 
but then $(v, b), (b, c), (c, v) \in F$ implies $(b, c) \in F_v$, 
a contradiction. 
\par
Suppose $(b, c) \in F_v^{-1}$ and $(a, b) \notin F_v^{-1}$. 
We have $(v, c), (c, b), (b, v), (a, b) \in F$ and $(c, a) \in \overline{F}$. 
We have from Claim~\ref{fact:C3-diamond} that 
there is a vertex $z$ such that $(b, z), (z, c) \in F$ and 
either $(v, z) \in \overline{F}$ or $(z, v) \in \overline{F}$. 
We first suppose $(v, z) \in \overline{F}$. 
If $(a, v) \in \overline{F}$, then 
$(c, a) \in \overline{F}$ implies $(c, v) \in \overline{F}$, 
contradicting $(v, c) \in F$. 
If $(v, a) \in \overline{F}$, then $(a, b, v)$ is 
a $\Delta$-obstruction, a contradiction. 
Thus $va \in E(G)$. 
We now have from Claim~\ref{fact:diagonal} that 
$(v, c, z, a)$ is a chordless cycle of length 4. 
Since $(v, c), (z, c) \in F$, we have $(v, a), (z, a) \in F$, 
but then $(v, a), (a, b), (b, v) \in F$ implies $(a, b) \in F_v$, 
a contradiction. 
We next suppose $(z, v) \in \overline{F}$. 
If $(a, z) \in \overline{F}$, then 
$(c, a) \in \overline{F}$ implies $(c, z) \in \overline{F}$, 
contradicting $(z, c) \in F$. 
If $(z, a) \in \overline{F}$, then $(a, b, z)$ is 
a $\Delta$-obstruction, a contradiction. 
Thus $za \in E(G)$. 
We now have from Claim~\ref{fact:diagonal} that 
$(z, c, v, a)$ is a chordless cycle of length 4. 
Since $(v, c), (z, c) \in F$, we have $(v, a), (z, a) \in F$, 
but then $(v, a), (a, b), (b, v) \in F$ implies $(a, b) \in F_v$, 
a contradiction. 
\end{proof}

\section{Concluding remarks}\label{section:conclusion}
In this paper, we provided a new algorithm for the recognition 
of simple-triangle graphs to improve the time bound 
from $O(n^2 \overline{m})$ to $O(nm)$. 
The algorithm uses the vertex ordering characterization 
in our previous paper~\cite{Takaoka18-DM} that 
a graph is a simple-triangle graph if and only if 
there is a linear ordering of the vertices 
containing both an alternating orientation of the graph and 
a transitive orientation of the complement of the graph. 
The algorithm finds such a vertex ordering 
or report that the given graph is not a simple-triangle graph. 
Correctness of the algorithm is due to 
the following structural characterization of simple-triangle graphs 
(Theorem~\ref{theorem:correctness}): 
a graph $G$ is a simple-triangle graph if and only if 
there is a pair of an alternating orientation $F$ of $G$ and 
a transitive orientation $\overline{F}$ of 
the complement $\overline{G}$ of $G$ such that 
$F \cup \overline{F}$ contains no $\Delta$-obstructions. 
We also showed, as a byproduct, that 
the recognition problem is NP-complete for acyclic alternately orientable graphs 
(Theorem~\ref{theorem:acyclic-alternately-orientable}), and 
alternately orientable cocomparability graphs can be recognized 
in $O(nm)$ time (Theorem~\ref{theorem:alternately-orientable-cocomparability}). 

Finally, we list some open problems related to the results of this paper: 
\begin{itemize}
\item Is there a recognition algorithm for simple-triangle graphs 
with time complexity less than $O(nm)$?
\item Is there a polynomial time algorithm for the isomorphism problem 
for simple-triangle graphs~\cite{Takaoka15-IEICE,Takaoka18-DM,Uehara14-DMTCS}?
\item Is there a forbidden characterization for simple-triangle graphs? 
\end{itemize}
